\pdfoutput=1

\documentclass{article}

\usepackage{times}

\usepackage[utf8]{inputenc}
\usepackage[american]{babel}
\usepackage{amssymb}
\usepackage{amsmath}
\usepackage{amsthm}
\usepackage{listings}
\usepackage{graphicx}
\usepackage{color}
\usepackage{subfig}
\usepackage{xspace}
\usepackage{enumerate}
\usepackage{xargs}
\usepackage{mathtools} 
\usepackage{algorithm}

\usepackage{algpseudocode}

\newtheorem{theorem}{Theorem}[section]
\newtheorem{lemma}[theorem]{Lemma}

\newtheorem{definition}[theorem]{Definition}

\newcommandx*{\bigO}[2][1=@pkling_false]{\mathcal{O}\ifthenelse{\equal{#1}{small}}{\bigl(#2\bigr)}{\left(#2\right)}}
\newcommandx*{\LDAUomicron}[2][1=@pkling_false]{\mathrm{o}\ifthenelse{\equal{#1}{small}}{\bigl(#2\bigr)}{\left(#2\right)}}
\newcommandx*{\LDAUOmega}[2][1=@pkling_false]{\Omega\ifthenelse{\equal{#1}{small}}{\bigl(#2\bigr)}{\left(#2\right)}}
\newcommandx*{\LDAUomega}[2][1=@pkling_false]{\omega\ifthenelse{\equal{#1}{small}}{\bigl(#2\bigr)}{\left(#2\right)}}
\newcommandx*{\LDAUTheta}[2][1=@pkling_false]{\Theta\ifthenelse{\equal{#1}{small}}{\bigl(#2\bigr)}{\left(#2\right)}}

\newcommandx*{\set}[2][2=@pkling_false]{\left\{#1\ifthenelse{\equal{#2}{@pkling_false}}{}{\;\middle|\;#2}\right\}}

\makeatletter
\renewcommand*{\@fnsymbol}[1]{}
\makeatother

\algblockdefx{ForPar}{EndForPar}[1]%
{\textbf{for }#1 \textbf{do in parallel}}%
{\textbf{end for}}

\makeatletter
\xspaceaddexceptions{\check@icr}
\makeatother

\begin{document}

\title{Online Top-$k$-Position Monitoring of Distributed Data Streams}

\author{Alexander M\"acker \and Manuel Malatyali\thanks{This work was partially supported by the German Research Foundation (DFG) within the Priority Program ``Algorithms for Big Data'' (SPP 1736) and by the EU within FET project MULTIPLEX under contract no.\ 317532.} \and Friedhelm Meyer auf der Heide \\ [0.4em]
	Heinz Nixdorf Institute \& 	Computer Science Department\\
	University of Paderborn, Germany\\[0.2em]
	\{amaecker, malatya, fmadh\}@hni.upb.de}

\date{}

\maketitle

\begin{abstract}
		Consider $n$ nodes connected to a single coordinator.
		Each node receives an individual online data stream of numbers and, at any point in time, the coordinator has to know the $k$ nodes currently observing the largest values, for a given $k$ between $1$ and $n$. 
		We design and analyze an algorithm that solves this problem while bounding the amount of messages exchanged between the nodes and the coordinator.
		Our algorithm employs the idea of using filters which, intuitively speaking, leads to few messages to be sent, if the new input is ``similar'' to the previous ones.
		The algorithm uses a number of messages that is on expectation by a factor of $\bigO{(\log \Delta + k)\cdot \log n}$ larger than that of an offline algorithm that sets filters in an optimal way, where $\Delta$ is upper bounded by the largest value observed by any node.
\end{abstract}

\section{Introduction}
In this paper we consider a model in which there is one coordinator and a set of $n$ distributed nodes connected to the coordinator. 
Each node continuously receives data from an input stream only known to the respective node or, in other words, observes a private function whose value changes over time. 
At any time, the coordinator has to know the $k$ nodes currently observing the $k$ largest values. In order to inform the coordinator about its current value, a node can exchange messages with the coordinator. 
Additionally, the coordinator can send broadcast messages received by all nodes. 
The goal in designing an algorithm for this setting, which we call \emph{Top-$k$-Position Monitoring}, is to find a solution that, on the one hand, keeps the coordinator informed as much as necessary for solving the problem and, at the same time, aims at minimizing the communication, i.e., the number of messages, between the coordinator and the distributed nodes. 

Such a model might be of interest in distributed settings.
Think of a set of sensors which can communicate directly to the coordinator in order to continuously keep track of the subset of $n$ locations at which current the highest $k$ values (of any parameter like speed, temperature, frequency, \ldots) are observed. 
Since these values will change over time generating an extremely large amount of data, it is not reasonable (and even not necessary) to keep the coordinator informed about all the changes observed somewhere in the system. 
Instead, reducing the communication is desired and might even get inevitable in large systems. 

For the considered problem, we present an algorithm that combines the notion of filters with a kind of random sampling of nodes. 
The basic idea of assigning filters to the distributed nodes is to reduce the number of exchanged messages by providing nodes constraints defining when they can safely resign to send observed changes in their input streams to the coordinator. 
However, if it might become necessary to communicate observed changes and update filters, we make extensive use of a new randomized protocol for determining the maximum (or minimum) value currently observed by (a certain subset of) the nodes.

As our problem is an online problem, since the values observed by the nodes change over time and are not known in advance, in our analysis we compare the number of messages exchanged by our online algorithm to that of an offline algorithm that sets filters in an optimal way. 
We show that they differ by a factor of at most $\bigO{(\log{\Delta} + k) \cdot \log{n}}$ on expectation, where $\Delta$ is the largest difference of the values observed at the nodes holding the $k$-th and $(k+1)$-st largest value at any time.

\subsection{Related Work}
\label{sec:relatedWork}
Efficient computation of functions on big datasets in terms of (single or distributed) data streams has turned out to be an important topic of research with applications in IP network traffic analysis, text mining or databases \cite{sanders}. 
The theory of data stream algorithms \cite{muthu} investigates the possibilities and limitations of approaches in this setting. 
The design of algorithms typically focuses on the computation of functions or their relaxations with objectives like using as less memory as possible, usually polylogarithmic in the length of the stream. 

Cormode et al.\ introduce the Continuous Monitoring Model \cite{cormodeSurvey} focusing on systems comprised of a coordinator and $n$ nodes generating or observing \emph{distributed} data streams. 
Consequently, the goal shifts to continuously (that is, in every time step) compute a function depending on the information available across all $n$ data streams up to the current time at a dedicated coordinator. 
The major objective becomes the minimization of the overall number of messages exchanged between the nodes and the coordinator.
We refer to this model and enhance it by a broadcast channel as proposed by Cormode et al.\ in \cite{cormodeFunction}. 

An important class of problems in this model investigated in literature are threshold computations, i.e., the coordinator is supposed to decide whether the current function value has reached some given threshold $\tau$ or not. If the function is monotone, e.g., monitoring the number of distinct values or the sum over all values, exact characterizations including matching lower bounds in the deterministic case are known \cite{cormodeFunction,cormodeSampling}. 
However, non-monotone functions, e.g., the entropy \cite{chakrabarti}, turned out to be much more complex to handle.

A general approach to reduce the communication when monitoring distrib\-uted data streams is proposed in \cite{zhangFilter}. Zhang et al.\ introduce the notion of \emph{filters}, which are also an integral part of our algorithm presented in this paper. 
They consider a problem called continuous skyline maintenance, in which a coordinator is supposed to continuously maintain the skyline of dynamic objects. 
As they aim at minimizing the communication overhead between the coordinator and the objects, they use a filter method that helps in avoiding the transmission of updates from the objects to the coordinator in case these updates cannot influence the skyline maintained by the coordinator. More precisely, the dynamic objects are points of a $d$-dimensional space and filters are hyper-rectangles assigned by the coordinator to the objects in such a way that as long as these points are within the assigned hyper-rectangle, updates need not be communicated to the coordinator. 

In their work \cite{yi}, Yi and Zhang were the first to study streaming algorithms with respect to their competitiveness. In their model there is one node and one coordinator and the goal is to keep the coordinator informed about the current value of a multivalued function $f: \mathbb{Z}^+ \to \mathbb{Z}^d$ that is observed by the node and changes its value over time, while minimizing the number of messages. 
Among others, for $d=1$, Yi and Zhang present an algorithm that is $\bigO{\log \delta}$-competitive if the last value received by the coordinator might deviate by $\delta$ from the current value of $f$. For arbitrary $d$, a competitiveness of $\bigO{d^2 \log(d \cdot \delta)}$ is shown.

Following the idea of studying competitive algorithms for monitoring streams and the notion of filters, Lam et al.\ \cite{lam} present an algorithm for online dominance tracking of distributed streams.
In this problem a coordinator always has to be informed about the dominance relationship between $n$ distributed nodes each observing an online stream of $d$-dimensional values. 
Here, the dominance relation of two nodes $b_1$ and $b_2$ currently observing $(p_1, \ldots, p_d)$ and $(q_1, \ldots, q_d)$, respectively, both from $\{1,\ldots, U\}^d$, is defined as $b_1$ dominates $b_2$ if $p_i < q_i$ for some $1 \leq i \leq d$ and $p_j \leq q_j$ for all other $1 \leq j \leq d$. 
Their algorithm is based on the idea of assigning filters to the nodes and they show that a mid-point strategy, which basically sets filters to be the mid-point between neighboring nodes, is $\bigO{d \log U}$-competitive with respect to the number of messages sent in comparison to an offline algorithm that sets filters optimally. 

A further problem related to ours, is examined in \cite{babcock}. 
In their work, Babcock and Olston consider the problem of distributed top-$k$ monitoring.
In the considered setting there is a set of objects $\{O_1, \ldots, O_n\}$ each associated with a numeric value. Additionally, there is a set of nodes each monitoring a stream of tuples $(O_i, \Delta)$ describing that the value associated with the object $O_i$ changes by $\Delta$. 
Note that these changes may occur at different nodes.
The goal then is to always keep the coordinator informed about the objects having the $k$ largest associated values.
If each node observes exactly one object and each object is observed by exactly one node, their problem is basically monitoring the $k$ largest values. 
By experimentally evaluating their approach, Babcock and Olston show that the amount of communication is an order of magnitude lower than that of a naive approach.

Finally, a model related to our (sub-)problem of finding the $k$-th largest value hold by the nodes exploiting the existence of a broadcast channel, was studied in literature as distributed selection using the shout-echo principle (see e.g.\ \cite{marberg,rotem}). However, this line of research aims at minimizing the number of communication cycles (i.e., the number of rounds each consisting of one broadcast and replies from all nodes), which is fundamentally different from our objectives.

\subsection{Our Contribution} 
We present and analyze a new randomized online algorithm for the Top-$k$-Position Monitoring problem, i.e., monitoring the nodes currently holding the $k$ largest values in a setting comprised of a coordinator and $n$ distributed nodes (Sect.~\ref{sec:model}). 

Our algorithm basically consists of two parts:
First of all, we use the notion of filters, which helps in avoiding the exchange of unnecessary updates between the nodes and the coordinator.
Second, whenever filters might be updated, a protocol, which uses $M(n)$ messages, to determine the maximum (or minimum) value hold by the distributed nodes is employed. 
We show that the amount of communication is by a factor of $\bigO{(\log{\Delta}+k)\cdot M(n)}$ larger than that of an offline algorithm which sets the filters optimally (Sect.~\ref{sec:filterAlg}).

We also present a concrete randomized protocol that determines the maximum value hold by $n$ nodes using $M(n) = \bigO{\log{n}}$ messages, with high probability, which we also prove to be asymptotically optimal (Sect.~\ref{sec:maxProtocol}). 
This result can be used to achieve a factor of $\bigO{(\log{\Delta}+k)\cdot \log n}$.

\section{Model and Techniques}
\label{sec:model}
In our setting there are $n$ distributed nodes identified by unique identifiers (IDs) from the set $\{1, \ldots, n\}$, each receiving a continuous data stream $(v_{i}^1, v_i^2, v_i^3 \ldots)$, which can be exclusively observed by the respective nodes. 
Also, at time $t$, a node $i$ observes $v_i^t \in \mathbb{N}$ and does not know any $v_i^{t'}$, $t' > t$.
We omit the index $t$ if it is clear from the context. 

The nodes can communicate to the coordinator, but cannot communicate to each other. 
They are able to store a constant number of integers, compare two integers and perform Bernoulli trials with success probability ${2^i}/{n}$ for $i \in \{0, \ldots, \log{n}\}$.
The coordinator can, on the one hand, communicate to one device, and has, on the other hand, a broadcast-channel to communicate one message received by all nodes at the same time. 
All the communication methods described above incur unit communication cost per message, the delivery is instantaneous, and we allow a message at time $t$ to have a size at most logarithmic in $n$ and $\max_{1\leq i\leq n}(v_i^t)$. 
Following the model in \cite{cormodeFunction}, we also allow that between any consecutive observations, i.e., between any two time steps $t$ and $t+1$, a communication protocol exchanging messages between the coordinator and the nodes may take place.

We consider a problem, which we call \emph{Top-$k$-Position Monitoring}, of monitoring the node IDs of the nodes holding the $k$ largest values, i.e., at each time $t$ the coordinator has to know the value of a function $\mathcal{F}$ mapping from the current values $v_1^t, \ldots, v_n^t$ to the set $\{s_1, \ldots, s_k\}$, $k \leq n$, of nodes currently observing the $k$ largest values among the (multi-)set $\{v_1^t, \ldots, v_n^t\}$. 
We refer to the set of nodes in $\{s_1, \ldots, s_k\}$ at time $t$ as top-$k$, explicitly the values at times $t'<t$ do not matter for this output. 
To ease the presentation, we assume that all $v_i$ are pairwise distinct at every time $t$. 
However, the main results remain valid if we drop this assumption.

In the following we assume w.l.o.g.\ that when talking about a fixed time $t$, node $i$ is at the $i$-th position in the ordered sequence of nodes, $\forall i < j \in \{1, \ldots, n\}:  v_i \geq v_j$, i.e., nodes are ordered descending. 

\subsection{A First Approach \& Classical Analysis}
One naive approach to monitor the Top-$k$-Positions, is to send each value observed by a node to the coordinator. 
In this case, the coordinator gets all the information necessary to determine the value of the function $\mathcal{F}$ at any time. 
However, this approach leads to a huge amount of messages and thus, in the following we are interested in more sophisticated solutions that aim at minimizing the number of messages exchanged between the coordinator and the distributed nodes. 

Now assume there is an algorithm that determines the maximum value observed by the $n$ nodes (Sect.~\ref{sec:maxProtocol}) using $\bigO{\log n}$ messages on expectation. 
Furthermore, assume this algorithm can be extended to determine the nodes within Top-$k$ using $\bigO{k \cdot \log n}$ messages on expectation. 
If we use this approach in each round to determine the Top-$k$, applying it for $T$ rounds yields $\bigO{T \cdot k \cdot \log n}$ messages. 

Determining the node holding the maximum value needs at least $\Omega(\log n)$ messages on expectation (Sect.~\ref{sec:maxProtocol}), implying that each randomized algorithm that determines the nodes within the top-$k$ needs at least $\Omega(\log n)$ messages on expectation.
Therefore, the algorithm described above is optimal up to a factor of $k$ for worst case inputs, i.e. inputs where the position of the maximum changes considerably from round to round.

However, we feel that this classical analysis does not reflect the quality of the algorithm properly for every instance.
Although the analysis states that this algorithm is almost optimal - on instances in which the new observed values are ``similar'' to the values observed in the last round, it behaves poorly.
On theses instances, the algorithm calculates the top-$k$ from scratch in every round, generating communication not needed by a ``somehow clever'' algorithm which, intuitively speaking, identifies and exploits these similarities.

To address this fact, in the next section we consider algorithms that are based on a concept called \emph{filters} \cite{zhangFilter} which aims at minimizing the communication by identifying changes in the observed values that need not be communicated to the coordinator. 

\subsection{Filter-Based Algorithms \& Competitive Analysis}
A set of filters is a collection of intervals, one assigned to each node, such that as long as the observed values at each node are within the given interval, the value of the function $\mathcal{F}$ does not change. 
For the problem at hand, this general idea of filters translates to the following definition.
\begin{definition}
	For a fixed time $t$, a \emph{set of filters} is an n-tuple of intervals $(F_1^t, \ldots, F_n^t)$, $F_i \subseteq \mathbb{N} \cup \{- \infty, \infty \}$ and $v_i \in F_i$, 
	such that as long as the value of node $i$ only changes within its interval, i.e., it holds $v_i \in F_i$, the value of the function $\mathcal{F}$ does not change.
	We use $F_i^t = \left[l_i^t, u_i^t \right]$ to denote the lower and upper bound of a filter interval, respectively.
\end{definition}

In our model, we assume that nodes are assigned such filters by the coordinator. 
If a node \emph{violates} its filter, i.e., the currently observed value is not contained in its filter, the node may report the violation and its current value to the coordinator. The coordinator then computes a new set of filters and sends them to the affected nodes. 
Thus, for such an algorithm to be correct, we demand that the intervals assigned to the nodes always constitute a (valid) set of filters and we call such an algorithm \emph{filter-based}. 
Note that the easiest way of defining a set of filters is to assign the value currently observed by a node as its interval. In this case the usage of filters does not lead to any benefit, so in general we are looking for filters that are as large as possible to minimize the count of filter changes which is directly related to the number of exchanged messages.

In order to serve its purpose that the value of the function $\mathcal{F}$ does not change as long as the observed values are within their filters, 
each pair of filters $\left(F_i,F_j\right)$ of nodes $i \in$ top-$k$ and $j \notin$ top-$k$ must be disjoint except for a (possible) single common point at their boundaries. Formally, we can state this observation as shown in the lemma below.

\begin{lemma}
	\label{le:filterDef}
	For a fixed time $t$, an n-tuple of intervals forms a set of filters if and only if
	\begin{enumerate}
		\item for all $i \in $ top-$k$ it holds $v_i \in F_i = [l_i,u_i]$ and $l_i \geq u_{j} \enspace \forall k < j \leq n$, and 
		\item for all $i \notin$ top-$k$ it holds $v_i \in F_i = [l_i,u_i]$ and $u_i \leq l_{j} \enspace \forall 1 \leq j \leq k$.
	\end{enumerate}
\end{lemma}

To analyze the quality of our online-algorithm, we use analysis based on competitiveness and compare the communication volume induced by the algorithm to that of an appropriately defined offline algorithm.
In our model, an offline algorithm knows all the input streams in advance, and without further restrictions can trivially solve the Top-$k$-Monitoring Problem without any communication.
In order to still get meaningful results and assess the quality of our algorithm, we assume the optimal offline algorithm $OPT$ to use filters assigned by the coordinator to the nodes. 
Hence, to lower bound the cost induced by $OPT$, we will essentially count the number of filter updates over time.  

We call an online algorithm $A$ \emph{c-competitive} if for every instance (i) its communication volume is by a factor of at most $c$ larger than the communication volume of $OPT$, in case $A$ is deterministic and (ii) the expected communication volume is by a factor of at most $c$ larger than the communication volume of $OPT$, in case $A$ is randomized.

\section{Filter-Based Algorithm for Top-$k$-Position Monitoring}
\label{sec:filterAlg}
In this section we describe our algorithm for monitoring the IDs of the nodes holding the $k$ largest values. 
Assuming that we have distributed protocols to determine the maximum and minimum value hold by the nodes (Sect.~\ref{sec:maxProtocol}), the key idea is to use filters to prevent that nodes communicate changes to the coordinator that are not essential for the correct computation of the function $\mathcal{F}$. 
In case filters are violated, these protocols are applied to gain information required for updating the set of filters properly. 

\subsection{Monitoring of Top-$k$-Position}
One solution to the Top-$k$-Position Monitoring problem is to use the online dominance tracking algorithm by Lam et al.\ \cite{lam} as discussed in Section~\ref{sec:relatedWork}.
However, it would no longer provide a $c$-competitive algorithm for any $c$. 
This is due to the fact that a lot of messages might be sent because of changing values of nodes that do not lead to a change in top-$k$ and thus are not sent by an optimal algorithm.

In Algorithm~\ref{alg:topk} we present a competitive solution for the Top-$k$-Position Monitoring problem, which adopts some basic idea of \cite{lam}.
At the beginning, the coordinator determines the $k$-th and $(k+1)$-st largest value and, based on their midpoint $M$, all filters are set to be $[-\infty, M]$ and $[M, \infty]$, respectively, due to the idea of Lam et al. As soon as nodes violate their filters, \emph{these} nodes apply the \textsc{MinimumProtocol} or \textsc{MaximumProtocol}, respectively. Hereby the coordinator gets informed that some filter violations have happened and based on the received values the coordinator can determine whether the top-$k$ set has changed or not. 
In the former case, the top-$k$ set is updated by repeatedly applying the \textsc{MaximumProtocol}. 
In the latter case, the coordinator determines the new midpoint between the $k$-th and $(k+1)$-st largest value and the nodes' filters are updated accordingly.

\begin{algorithm}
	\caption{Top-$k$-Position Monitoring}
	\label{alg:topk}
 
\begin{algorithmic}[1]
	\State FilterReset();
	\Comment{time $t = 0$ initialization of filters}	
	\item[]
	\ForPar{Node $i \in \{1, \ldots, n\}$ }	
	\Comment{let $t > 0$ be the current time} 
		\If{ Node $i$ has a filter violation} 
			\If{ Node $i$ was in top-$k$ at time $t$-1} 
				\State \textsc{MinimumProtocol}($k$);
			\Else
				\State \textsc{MaximumProtocol}($n-k$);
			\EndIf
		\EndIf
	\EndForPar
	\item[]
	
	\Comment{ coordinator executes the following lines, including the procedures }
	\If{ at least one node has communicated a value at the end \\ ~~~~ of \textsc{MinimumProtocol} in line 5 or \textsc{MaximumProtocol} in line 7 }
		\State \textsc{FilterViolationHandler}();
	\EndIf
	\item[]
	\Procedure{\textsc{FilterViolationHandler()}}{}
	\If{\textsc{MinimumProtocol} has communicated a minimum to the coordinator} 
		\State $min \gets$ minimum communicated in line 5;
	\EndIf
	
	\If{\textsc{MaximumProtocol} has communicated a maximum to the coordinator}
		\State $max \gets$ maximum communicated in line 7;
	\EndIf

	\If{$max$ not set}
		\State $max \gets$ apply \textsc{MaximumProtocol}($n-k$) to all $i \notin $ top-$k$ at time $t-1$;
	\Else
		\State $min \gets$ apply \textsc{MinimumProtocol}($k$) to all $i \in $ top-$k$ at time $t-1$;
	\EndIf
	
	\State $T^+(t_0, t) \gets \min ( T^+(t_0, t-1), min)$;
	\State $T^-(t_0, t) \gets \max ( T^-(t_0, t-1), max)$;
	
	\If{ $T^+(t_0,t) < T^-(t_0,t)$ }
		\State \textsc{FilterReset}();
	\Else 
		\State $M \gets $ midpoint of the interval $[T^+(t_0,t), T^-(t_0,t)]$;
		\State broadcast $M$ to let the nodes update their filters;
	\EndIf
	\EndProcedure
	\item[]

	\Procedure{\textsc{FilterReset()}}{}
		\For{$i \gets 1$ to $k+1$} 
			\State apply \textsc{MaximumProtocol}($n$) to all nodes $i' \notin $ top-$(i-1)$;
		\EndFor
		\State $M \gets $ midpoint of the interval between the k-th and (k+1)-st largest value;
		\State broadcast $M$ to let the nodes $i' \in$ top-$k$ reset their filters to $[M, \infty]$ and $i' \notin$ top-k analog;
	\EndProcedure
\end{algorithmic}
\end{algorithm}

In order to prove the competitiveness of Algorithm~\ref{alg:topk}, we need the following definition and the consecutive Lemma~\ref{le:filterChange}:
\begin{definition}
	Let $t_0, t$ be given times with $t \geq t_0$. 
	We denote the maximum over all values observed by nodes $i \notin$ top-$k$ during the interval $[t_0,t]$ by $T^-(t_0, t) = \max_{t_0 \leq t' \leq t} \max_{i \in \{k+1, \ldots, n\} }  ( v_i^{t'} )$.
	Consequently, let $T^+(t_0, t)$ be the minimum over all values observed by nodes $i \in $ top-$k$ during the interval $[t_0,t]$.
\end{definition}

\begin{lemma}
	\label{le:filterChange}
	If $OPT$ uses the same set of filters throughout the interval $[t', t'']$, the minimum over all nodes $i \in $ top-$k$ is greater or equal the maximum over all nodes $i \notin$ top-$k$ during $[t', t'']$. Formally, $T^+(t', t'') \geq T^-(t', t'')$. 
\end{lemma}

\begin{proof}
	Proof by contradiction. 
	Assume $OPT$ uses the same set of filters throughout the interval $[t', t'']$, but \linebreak $T^+(t', t'') < T^-(t', t'')$ holds. 
	Then there are two nodes, $i \in$ top-$k$ and $j \notin$ top-$k$, and two times $t_1, t_2 \in [t', t'']$, such that $v_i^{t_1} = T^+(t', t'')$ and $v_j^{t_2} = T^-(t', t'')$.
	Due to the definition of a set of filters and the fact that $OPT$ has not communicated during $[t', t'']$, $OPT$ must have set the filter for node $i$ to $[s_1, \infty]$, $s_1 \leq v_i^{t_1}$, and for node $j$ to $[-\infty, s_2]$, $s_2 \geq v_j^{t_2}$.
	This is a contradiction to the definition of a set of filters and Lemma~\ref{le:filterDef}.
\end{proof}

The results above lead to the following theorem.

\begin{theorem}
	\label{th:competitiveness}
	Let $M(n)$ be the number of messages required to determine the maximum or minimum value by \textsc{MaximumProtocol} and \textsc{MinimumProtocol} in Algorithm~\ref{alg:topk}.
	Furthermore, let $\Delta = max_{t} (v_k^t-v_{k+1}^t)$. Then, Algorithm~\ref{alg:topk} is $\bigO{(\log \Delta + k) \cdot M(n)}$-competitive.
\end{theorem}
\begin{proof}
	To prove the theorem, we split the time into intervals based on the time steps $t_1, t_2, \ldots$ at which $OPT$ communicates.
	Let $t_0$ be the time step directly after the initialization.
	The outline of our proof is as follows: We consider an arbitrary interval $I := [t_{i-1}+1, t_i-1]$ for $i > 1$.
	First of all, we count the number of executions of the different procedures of our algorithm during this interval $I$. Then, we upper bound the number of messages exchanged during an execution of these procedures and finally, combine both observations to upper bound the overall number of messages sent during the interval $I$.
	
	First, observe that $OPT$ does not communicate during $I$ and thus, the top-$k$ set does not change.
	According to Lemma~\ref{le:filterChange} the procedure \textsc{FilterReset} is executed at most once during the interval $I$ since it is only called if $T^+(t_0,t) < T^-(t_0,t)$.
	Furthermore, at the beginning of $I$ the difference between $T^+(\cdot, t_{i-1}+1)$ and $T^-(\cdot, t_{i-1}+1)$ is at most $\Delta$.
	Because the procedure \textsc{FilterViolationHandler} is only called when a filter-violation occurs and due to the fact that in lines $23-25$ the difference between $T^+$ and $T^-$ is at least halved, there can be at most $\log \Delta$ time-steps in which filter-violations occur between two calls of \textsc{FilterReset}. Hence, during the interval $I$ there can be at most $2\cdot \log \Delta$ calls of the procedure \textsc{FilterViolationHandler}.
	
	Next we upper bound the number of messages sent by the algorithm during the interval $I$.
	Fix a time $t$. If at time $t$ there is no filter-violation, the algorithm incurs no communication.
	In case there are filter-violations, lines $5$ and/or $7$ may need $M(n-k)+M(k)$ messages.
	Additionally, a call of the procedure \textsc{FilterViolationHandler} may lead to the exchange of $\bigO{M(n)}$ messages and the procedure \textsc{FilterReset} to $(k+1)\cdot M(n)$ messages.
	In conclusion, the algorithm needs at most $\bigO{(\log \Delta + k) \cdot M(n))}$ messages in the interval $I$.
	Observe that $OPT$ communicates at time $t_i$ and our algorithm may communicate at most $k \cdot M(n)$ messages at $t_i$.
	
	Thus, if we consider the execution of the algorithm for a certain time horizon and assume that $OPT$ communicates at least $r$ times, Algorithm~\ref {alg:topk} communicates at most $(r+1) \cdot \bigO{(\log \Delta + k)\cdot M(n)}$ messages and hence, the algorithm is $\bigO{(\log \Delta + k) \cdot M(n)}$-competitive.
\end{proof}

\section{Distributed Protocol For Maximum Computation}
\label{sec:maxProtocol}
In this section we present and analyze a protocol to determine the maximum value currently hold by a set of nodes. Here we no longer consider an online problem, as we assume that the protocol is applied at a \emph{fixed} time $t$. Hence, also the values of the nodes applying the protocol do not change during a single execution.

\begin{algorithm}
	\caption{MaximumProtocol}
	\label{alg:max}
	\begin{algorithmic}[1]
	\Procedure{MaximumProtocol}{$N \in \{n, n+1, \ldots \}$} 
	\Comment{$N$ is an upper bound on number of nodes}	
	\ForPar{Node $i \in \{1, \ldots, n \}$}
	\State $active \gets$ true;
	\EndForPar

	\For{Round $r \gets  0$ to $\log N$}
		\ForPar{Node $i \in \{1, \ldots, n \}$}
			\If{Node $i$ is active}
				\If{$max_{r-1} > v_i$}
					\State goto line 14;
				\EndIf		
				\State $p \gets$ Result of a coin flip with success-probability $\frac{2^{r}}{N}$;
				\If{$p = 1$}
					\State Send ($i$, $v_i$) to coordinator;
					\State $active \gets$ false;
				\EndIf
			\EndIf
			\EndForPar
		\State coordinator broadcasts maximum $max_r$ of all seen values $v_i$;
	\EndFor

\EndProcedure

	\end{algorithmic}
\end{algorithm}

The idea of how to determine the maximum value currently hold by $n$ nodes is as follows (c.f.\ Algorithm~\ref{alg:max}): 
The algorithm proceeds in $\log N$ rounds, for a given $N \geq n$. 
At the beginning, i.e., before round $r=0$, all nodes are set to be active. 
In each round $r$, each active node decides independently of the other nodes to send its value to the coordinator with a probability of $2^r/N$. 
After this, the coordinator broadcasts the largest value observed so far and nodes having a value not larger than this maximum are deactivated, i.e., no longer take part in the algorithm. 
This broadcast starts the next round which continues with all nodes still being active. 
Note that Algorithm~\ref{alg:max} constitutes a Las Vegas algorithm that always computes the correct solution in $\log N$ steps, but the number of messages exchanged is described by a random variable.

To analyze the expectation of this random variable, we first analyze the probability that a fixed node sends a message to the coordinator during one execution of \textsc{MaximumProtocol}.
\begin{lemma}
	\label{le:probabilitySending}
	Let $X_i$ be a binary random variable indicating whether node $i \in \{1, \ldots, n\}$ sends a message during a run of Algorithm 1.
	Then,
	\[ \Pr[X_i=1]  \leq \frac{1}{N} + \sum_{r=1}^{\log N} \frac{2^{r}}{N} \cdot \left( 1- \frac{2^{r-1}}{N}\right)^i. \]
\end{lemma}

\begin{proof}
	Let $S_{i,j}$ be a binary random variable indicating whether node $i$'s coin flip is successful in round $r$. Let $Y_{i,j}$ be a binary random variable indicating that node $i$ is active in round $j$. 
	A node sends a message in round $j$ if and only if it is active in round $j$ and its coin flip results in a success. 
	Furthermore it deactivates itself after sending a message, so it is not possible that one node sends two or more messages. 
	
	Hence, 
	\begin{align*}
		\Pr[X_i=1] &= \sum_{r=0}^{\log N} \Pr[ S_{i, r} \mbox{\textbf{ and }} Y_{i,r}=1 ] \\
		&= \sum_{r=0}^{\log N} \Pr[ S_{i,r} \big\vert Y_{i,r}=1 ]  \cdot \Pr[Y_{i,r} = 1] \\
		&= \sum_{r=0}^{\log N} \frac{2^{r}}{N} \cdot \Pr[ Y_{i,r}=1]\enspace.
	\end{align*}
	Observe that a node $i$ is active at the beginning of round $j$ if and only if it is active at the beginning of round $j-1$ and no node $i' \leq i$ sends a message in round $j-1$. Hence, we obtain $\Pr[Y_{i,0}=1] = 1$ and for $j > 0$:
	\begin{align*}
		&\Pr[Y_{i,j}] =
	\Pr [\mbox{no } i'\leq i \mbox{ sends in round } j-1 \big\vert Y_{i, j-1} = 1] 
		  \cdot \Pr [Y_{i, j-1}=1] \enspace .
	\end{align*}
	In a fixed round, the decisions of active nodes whether to send a message or not are independent. Additionally, $Y_{i,j} = 1$ implies $Y_{i',j} =1 \enspace \forall 1 \leq i' \leq i$. Hence, 
	\begin{align*}
		&\Pr [\mbox{no } i'\leq i \mbox{ sends in round } j-1 \big\vert Y_{i, j-1} = 1] 
		\leq  \left(1-\frac{2^{j-1}}{N}\right)^i
	\end{align*}
	and thus by an inductive argument, 
	\begin{align*}
		\Pr[Y_{i,j}=1] \leq \prod_{r=0}^{j-1} \left(1-\frac{2^{r}}{N} \right)^i \leq \left(1-\frac{2^{j-1}}{N} \right) ^i\enspace.
	\end{align*}
	Therefore the probability for $X_i =1$ is bounded as claimed by the lemma.  
\end{proof}
This result about the probability that a node sends a message can be used to determine an upper bound on the expected communication volume.

\begin{theorem}[Upper Bound]
	\label{th:max}
	The expected number of messages sent in Algorithm~\ref{alg:max} is at most $2\cdot \log N + 1$ and $\bigO{\log N}$ with high probability, i.e., for any fixed $c>1$ it is $\bigO{\log{N}}$ with probability at least $1-\frac{1}{N^c}$.
\end{theorem}

\begin{proof}
	By Lemma~\ref{le:probabilitySending}, we have $\mathbb{E}[X_i] \leq \frac{1}{N} + \sum_{r=1}^{\log N} \frac{2^{r}}{N} \cdot \left( 1- \frac{2^{r-1}}{N}\right)^i$.  Hence,
	\begin{align*}
		\mathbb{E}[X] &\leq 1 + \sum_{i=1}^n \sum_{r=1}^{\log{N}} \frac{2^{r}}{N} \left( 1- \frac{2^{r-1}}{N}\right)^i  
		\leq 1 + \frac{1}{N} \sum_{r=1}^{\log{N}} 2^{r} \sum_{i=0}^{n-1} \left( 1- \frac{2^{r-1}}{N}\right)^i \enspace.
	\end{align*}
	Replacing the terms of the geometric series, we obtain
	\begin{align*}
		 1 + \frac{1}{N} \sum_{r=1}^{\log{N}} 2^{r} \frac{1}{1-\left(1-\frac{2^{r-1}}{N} \right)}
		&= 1 + \frac{1}{N} \sum_{r=1}^{\log{N}} 2^{r} \frac{1}{\frac{2^{r-1}}{N}}
		= 1 + \sum_{r=1}^{\log{N}} 2 \leq 2\cdot \log{N}+1 \enspace ,
	\end{align*}
	proving the expected number of messages.
	
	Note that although the random variables $X_i$ are not independent, a variable $X_i$ only depends on those $X_j$ with $j<i$ and observing the event that a node $i$ sends a message, can only decrease the probability of sending a message of another node.
	It is known that Chernoff bounds for the upper tail of a distribution can be applied to variables satisfying such a kind of negative correlation \cite{chernoff}.
	More precisely, we can apply such a Chernoff bound if for all $I \subseteq \{1, \ldots, n\}$ it holds 
	\[
	\Pr[\forall i \in I: X_i=1] \leq \prod_{i \in I}\Pr[X_i=1]\enspace.
	\] 
	Without loss of generality assume that $I = \{i_1, \ldots, i_l\}$ and $v_{i_1} \geq v_{i_2} \geq \ldots \geq v_{i_l}$. Then we have
	\begin{align*}
		\Pr[\forall i \in I: X_i=1] 
		&= \Pr[X_{i_1} = 1 \wedge \ldots \wedge X_{i_l}=1] \\
		& = \Pr[X_{i_1}=1 \big\vert X_{i_2}=1 \wedge \ldots \wedge X_{i_l}=1]
		 \cdot \Pr[X_{i_2}=1 \wedge \ldots \wedge X_{i_l}=1]\\
		& \leq \Pr[X_{i_1}=1] \cdot \Pr[X_{i_2}=1 \big\vert X_{i_3}=1 \wedge \ldots \wedge X_{i_l}=1]
		 \cdot \Pr[X_{i_3}=1 \wedge \ldots \wedge X_{i_l}=1]\\
		& \leq \Pr[X_{i_1}=1] \cdot \Pr[X_{i_2}=1] \cdot \ldots \\
		& \hspace{1cm} \vdots \\
		&\leq \prod_{i \in I}\Pr[X_i=1]
	\end{align*}
	and thus, the $X_i$'s are negatively correlated and we obtain the claimed result by applying a standard Chernoff bound.
\end{proof}

We show that the bound of Theorem~\ref{th:max} is optimal, by showing that if all values hold by $n$ nodes are pairwise distinct, then any randomized algorithm needs a number of messages that is at least logarithmic, i.e., $\mathbb{E}[X] = \Omega(\log{n})$. By following Yao's minimax principle, it is sufficient to show that, given a probability distribution on the inputs, any deterministic algorithm has to send $\Omega(\log{n})$ messages on expectation. 

\begin{theorem}[Lower Bound]
	\label{th:yao}
	Every randomized algorithm requires $\Omega(\log{n})$ messages on expectation to compute the maximum in our model.
\end{theorem}

\begin{proof}
	Let the inputs be distributed according to a distribution $P$ in such a way that each instance where each of the numbers $\{1, \ldots, n\}$ is assigned to exactly one node $v_i$ is chosen with probability $(1/n!)$.
	Consider any deterministic algorithm $A$ that calculates the maximum in our model. Since we are looking for a lower bound on the number of messages, we see that $A$ can basically not do better than having a fixed sequence $(s_1, \ldots, s_n)$ of nodes that it probes consecutively in this ordering, skipping nodes that have values smaller than the maximum value observed so far. 
	
	By the following simple argument, we can show the desired result:
	Assume for the moment, that A probes \emph{all} nodes, i.e., A receives a randomly chosen permutation of the values $\{1, \ldots, n\}$. The course of the algorithm can be (partly) described by gradually constructing a binary search tree of the observed values. Now consider the actual course of the algorithm, i.e., A skips nodes which cannot deliver new information about the maximum value. Then the course of the algorithm is nothing else but the path in the binary search tree from the root to the node holding the maximum value. It is known (e.g., \cite{bst}) that on expectation this path has a length of $\Theta(\log{n})$, proving the theorem. 
\end{proof}
Conclusively, we can combine Theorem~\ref{th:competitiveness} and \ref{th:max} to obtain our main result.
\begin{theorem}
	Using Algorithm~\ref{alg:max} and its analog for determining the minimum in Algorithm~\ref{alg:topk}, yields a randomized $\bigO{(\log \Delta + k) \cdot \log{n}}$-competitive algorithm for Top-$k$-Position Monitoring.
\end{theorem}

\section{Summary}
In this paper we presented and analyzed an online algorithm for keeping track of the $k$ largest values currently observed at $n$ distributed data streams. 
Our analysis shows that the approach is $\bigO{(\log \Delta + k)\cdot \log n}$-competitive with respect to an optimal filter-based offline algorithm. 
Hence, it is particularly suited for applications in large systems where one is only interested in collecting information about a small (constant sized) subset of nodes. 
In contrast to settings where the observed values (and thus the competitiveness) may grow large, the approach performs quite well when these values are naturally bounded by the application domain. 
This might be, e.g., the case when monitoring environmental information like temperatures, frequencies and similar parameters. 

It is worth mentioning that for the optimal offline algorithm our analysis only depends on the number of filter updates the algorithm communicates. 
It might be interesting to also investigate the number of messages sent by the nodes to the coordinator in order to get stronger bounds on the optimal filter-based algorithm and thus, to further improve the competitiveness factor.

As an integral part of our algorithm, we examined the question of how to distributedly compute the maximum value hold by the nodes at a fixed time. 
It might turn out that this protocol is useful in other problems considered in the distributed monitoring model. 
For a variant of our Top-$k$-Position Monitoring problem in which one is not only interested in the top-$k$ set but also the ordering of these nodes according to their values, we conjecture that a combination of the approach by Lam et al.\ \cite{lam} and our protocol might lead to an $\bigO{\log \Delta \cdot \log(n-k)}$-competitve algorithm. Still, a detailed analysis remains open for future work.

\end{document}